\documentclass[11pt]{article}

\usepackage{amsmath,amsfonts,amssymb,amsthm}
\usepackage[margin=1in]{geometry}
\usepackage[colorlinks]{hyperref}
\usepackage{enumitem}
\usepackage{mathtools}
\usepackage{enumitem}
\usepackage{algorithm, caption}
\usepackage[noend]{algpseudocode}
\usepackage{subfigure}
\usepackage{graphicx}
\usepackage{tabularx,environ}
\usepackage[braket]{qcircuit}
\usepackage[dvipsnames]{xcolor}
\usepackage{multirow}
\usepackage[capitalize]{cleveref}
\usepackage{dsfont}

\makeatletter
\newcommand*\bigcdot{\mathpalette\bigcdot@{.5}}
\newcommand*\bigcdot@[2]{\mathbin{\vcenter{\hbox{\scalebox{#2}{$\m@th#1\bullet$}}}}}
\makeatother

    \hypersetup{
    	colorlinks,
    	linkcolor={red!100!black},
    	citecolor={blue!100!black},
    }
\makeatletter
\newcommand\notsotiny{\@setfontsize\notsotiny\@vipt\@viipt}
\makeatother

\newcommand{\norm}[1]{\left\|{#1}\right\|}

\def\01{\{0,1\}}

\newcommand{\eps}{\varepsilon}

\theoremstyle{plain}

\newtheorem{theorem}{Theorem}

\newtheorem{fact}[theorem]{Fact}


\definecolor{applegreen}{rgb}{0.0, 0.5, 0.0}

\newcommand{\Id}{\ensuremath{\mathop{\rm Id}\nolimits}}

\newcommand{\poly}{\mbox{\rm poly}}

\DeclareMathOperator{\Tr}{Tr}

\newcommand{\beq}{\begin{equation}}
\newcommand{\eeq}{\end{equation}}
\newcommand{\beqn}{\begin{equation*}}
\newcommand{\eeqn}{\end{equation*}}
\newcommand{\beqr}{\begin{eqnarray}}
\newcommand{\eeqr}{\end{eqnarray}}
\newcommand{\beqrn}{\begin{eqnarray*}}
\newcommand{\eeqrn}{\end{eqnarray*}}
\newcommand{\bmline}{\begin{multline}}
\newcommand{\emline}{\end{multline}}
\newcommand{\bmlinen}{\begin{multline*}}
\newcommand{\emlinen}{\end{multline*}}

\theoremstyle{plain}

\newtheorem{claim}[theorem]{Claim}

\theoremstyle{definition}

\newtheorem{problem}[theorem]{Problem}
\theoremstyle{remark}
\newtheorem{remark}[theorem]{Remark}

\renewenvironment{proof}[1][]{
    \begin{trivlist}
    \item[\hspace{\labelsep}{\em\noindent Proof#1:\/}]}
    {{\hfill$\Box$}
    \end{trivlist}
}

\newtheoremstyle{named}{}{}{\itshape}{}{\bfseries}{.}{.5em}{\thmnote{#3}}
\theoremstyle{named}

\usepackage{tikz}
\usepackage{thm-restate} 
\title{Simple algorithms to test and learn  local Hamiltonians}
\author{ 
Francisco Escudero Gutiérrez\thanks{Qusoft and CWI \href{feg@cwi.nl}{feg@cwi.nl} } \and}

\date{ \today }

\begin{document}

\maketitle
\begin{abstract}
    We consider the problems of testing and learning  an $n$-qubit $k$-local  Hamiltonian from queries to its evolution operator with respect the 2-norm of the Pauli spectrum, or equivalently, the normalized Frobenius norm. 
    For testing whether a Hamiltonian is $\eps_1$-close to $k$-local or $\eps_2$-far from $k$-local, we show that $O(1/(\eps_2-\eps_1)^{8})$ queries suffice. This solves two questions posed in a recent work by Bluhm, Caro and Oufkir. For learning up to error $\eps$, we show that $\exp(O(k^2+k\log(1/\eps)))$ queries suffice. Our proofs are simple, concise and based on Pauli-analytic techniques.
\end{abstract}
\section{Introduction}
In this work we consider the problems of testing and learning a local Hamiltonian from its time evolution operator. These kinds of Hamiltonians govern the dynamics of many physical systems, which motivates the problem of testing if an unknown Hamiltonian is local and, if it is local, to learn it in an efficient way. In fact, there is vast and recent literature about learning a local Hamiltonian from its time evolution operator \cite{Silva2011Practical, Bairey2019Learning, Zubida2021Optimal, haah2022optimal, wilde2022scalably, yu2023robust, caro2023learning, Dutkiewicz.2023, huang2023heisenberg, castaneda2023Hamiltonian, li2023heisenberglimited, möbus2023dissipationenabled, franca2024efficient, Gu2022Practical}, but the first testing algorithm was only proposed recently by Bluhm, Caro and Oufkir \cite{bluhm2024Hamiltonianv1}. There are also many of results about learning a Hamiltonian from its Gibbs state \cite{anshu2021sample, haah2022optimal, rouze2023learning, onorati2023efficient, bakshi2023learning, Gu2022Practical}, but in this work we will only consider the model where one accesses the Hamiltonian through queries to its evolution operator.

An $n$-qubit Hamiltonian $H$ is a self-adjoint operator acting on $(\mathbb C^{2})^{\otimes n}$. As such, it can be expanded in terms of the Pauli strings as 

\begin{equation*}
    H=\sum_{x\in \{0,1,2,3\}^n} h_x\sigma_x,
\end{equation*}
where $h_x$ are real numbers and $\sigma_x=\otimes_{i\in [n]}\sigma_{x_i},$ where $\sigma_0=\Id_2,\, \sigma_1=X,\,\ \sigma_2=Y,\, \sigma_3=Z$. The coefficients $h_x$ are known as the Pauli spectrum of the Hamiltonian. A Hamiltonian $H$ is $k$-local if it is supported on Pauli strings that act trivially on all but at most $k$ qubits. In other words, $H$ is $k$-local if $h_x=0$ for all $x\in\{0,1,2,3\}^n$ that take values different from $0$ on more than $k$ sites. If $H$ is the Hamiltonian describing the dynamics of a certain physical system, then the states of that system evolve according to the time evolution operator $U(t)=e^{-iHt}$. This means that if $\rho(0)$ is the state at time $0$, at time $t$ the state will have evolved to $\rho(t)=U(t)\rho(0)  U^{\dagger}(t).$ Hence, to test and learn a Hamiltonian one can do the following: prepare a desired state, apply $U(t)$ or tensor products of $U(t)$ with identity to the state (i.e., query $U(t)$), and finally measure in a chosen basis. It is usual to impose the normalization condition $\norm{H}_{\infty}\leq 1$ (i.e., that the eigenvalues of $H$ are bounded in absolute value by $1$), as otherwise the complexities scale with the norm of the Hamiltonian. Given a distance $d$ in the space of Hamiltonians, $H$ is $\eps$-far from being $k$-local if for every $k$-local $H'$ it satisfies that $d(H,H')>\eps,$ and otherwise is $\eps$-close. Now we are ready to state the testing and learning problems. 

\begin{problem}[Tolerant Hamiltonian locality testing]\label{prob:localitytesting}
    Let $0\leq \eps_1<\eps_2$, $\delta\in (0,1)$, $k\in \mathbb N$ and $d$ be a distance between Hamiltonians. Let $H$ be an $n$-qubit Hamiltonian with $\norm{H}_{\infty}\leq 1$ that is promised to be either $\eps_1$-close to $k$-local or $\eps_2$-far from $k$-local. The problem is to decide between those cases with success probability $\geq 1-\delta$  by making queries to $U(t)$. 
\end{problem}

\begin{problem}[Local Hamiltonian learning]\label{prob:locallearning}
    Let $\eps>0$, $\delta\in (0,1)$, $k\in \mathbb N$ and $d$ be a distance between Hamiltonians. Let $H$ be a $k$-local $n$-qubit Hamiltonian with $\norm{H}_{\infty}\leq 1$. The problem is to output a classical description of a $k$-local Hamiltonian $H'$ that is $\eps$-close to $H$ with probability $\geq 1-\delta$  by making queries to $U(t)$.
\end{problem}

In both problems, the main goals are minimizing the number of queries and the total evolution time.  These quantities depend on the distance $d$. The recent learning literature usually takes $d$ to be the supremum distance of the Pauli spectrum ($d(H,H')=\norm{H-H'}_{\mathrm{Pauli},\infty}=\max_x |h_x-h'_x|$), and imposes extra constraints on the Pauli spectrum (generally some kind of geometrical locality, such as assuming that the qubits are displayed in a grid and the Pauli terms only act on neighboring qubits). However, if one wants to convert these learning algorithms to learners under a finer distance such as the 2-norm of the Pauli spectrum ($d(H,H´)=\norm{H-H'}_2=\sqrt{\sum_x |h_x-h'_x|^2}$, which equals the normalized Frobenius norm due to Parseval's identity) it is not clear how to avoid a $\poly(n^k)$ term appearing, due to the fact that there are $n^k$ Pauli strings that are $k$-local. This 2-norm of the Pauli spectrum was recently considered by Bluhm, Caro and Oufkir \cite{bluhm2024Hamiltonianv1}, who proved that to learn an arbitrary $n$-qubit Hamiltonian under this distance it is necessary to make $\Omega(2^{2n})$ queries to the time evolution operator. In the first version of their work, they also proposed a non-tolerant testing algorithm, meaning that it only works for the case $\eps_1=0,$ whose query complexity is $O(n^{2k+2}/(\eps_2-\eps_1)^4)$ and with total evolution time $O(n^{k+1}/(\eps_2-\eps_1)^3)$. They posed as open questions whether the dependence on $n$ could be removed and whether an efficient tolerant-tester was possible \cite[Section 1.5]{bluhm2024Hamiltonianv1}. Our first result gives positive answer to both questions. From now on, unless otherwise mentioned, we will consider the distance $d$ to be the one induced by the 2-norm of the Pauli spectrum. 

\begin{theorem}[Locality testing]\label{theo:localitytesting}
    There is an algorithm that solves the locality testing problem (\cref{prob:localitytesting}) by making $O(1/(\eps_2-\eps_1)^8\cdot\log(1/\delta))$ queries to the evolution operator and with $O(1/(\eps_2-\eps_1)^7\cdot\log(1/\delta))$ total evolution time.  
\end{theorem}

Our algorithm to test for locality is simple. It consists of repeating the following process $1/(\eps_2-\eps_1)^8$ times: prepare $n$ EPR pairs, apply $U(\eps_2-\eps_1)\otimes \Id_{2^n}$ to them and measure in the Bell basis. Each time that we repeat this process, we sample from the  Pauli sprectrum of $U(\eps_2-\eps_1)$\footnote{The Pauli spectrum of a unitary $U=\sum_x u_x\sigma_x$ determines a probability distribution because $\sum_x |u_x|^2=1.$}. As $\eps_2-\eps_1$ is  small, Taylor expansion ensures that $U(\eps_2-\eps_1)\approx \Id_{2^n}-i(\eps_2-\eps_1) H$, so sampling from the Pauli spectrum of $U(\eps_2-\eps_1)$  allows us to estimate the weight of the non-local terms of $H.$ If that weight is big, we output that $H$ is far from $k$-local, and otherwise we conclude that $H$ is close to $k$-local.

After sharing \cref{theo:localitytesting} with Bluhm, Caro and Oufkir, they independently improved the analysis of their testing algorithm. They showed that in a certain broad regime their tester is tolerant, only makes $O(1/((\eps_2-\eps_1)^3\eps_2))$ queries and just requires $O(1/((\eps_2-\eps_1)^{2.5}\eps_2^{0.5}))$ total evolution time \cite[Theorem B.5]{bluhm2024hamiltonianv2}. To compare \cref{theo:localitytesting} with the improved testing result of Bluhm, Caro and Oufkir, we should introduce a family of problems that has locality testing as an instance. Given a subset $S$ of $\{0,1,2,3\}^n$, we define the problem of testing property $S$ as the problem of testing whether $H$ is $\eps_1$-close to be supported on $S$ or $\eps_2$-far from being supported on $S$. By taking $S$ equal to the set of strings that take the value $0$ on at least $n-k$ sites one recovers the $k$-locality testing problem. Both our \cref{theo:localitytesting} and their testing result work for testing properties defined by sets $S$. The advantages of their algorithm are that it uses no auxiliary qubits, while ours requires $n$, and  it is quadratically better than ours with respect to $\eps_2-\eps_1$. The advantages of our algorithm are that it works for any $S$, while theirs only works for $|S|=O(2^{n}(\eps_2^2-\eps_1^2)^{2})$, although they can remove this constraint by allowing access to $O(\log(|S|^3/2^n)+\log(1/(\eps_2-\eps_1)))$ auxiliary qubits; that our algorithm just needs $O(1/(\eps_2-\eps_1)^8)$ classical post-processing time, while theirs requires $O(n^2|S|/((\eps_2-\eps_1)^3\eps_2))$ time; and that theirs needs $H$ to be traceless, while ours does not. A technical feature of our result is that the proof is simpler and more concise than theirs.

Our second result is a learning algorithm for $k$-local Hamiltonians. 

\begin{theorem}[Local Hamiltonian learning]\label{theo:locallearning}
    There is an algorithm that solves the local Hamiltonian learning problem (\cref{prob:locallearning})  by making $\exp(O(k^2+k\log (1/\eps))\log(1/\delta)$ queries to the evolution operator with $\exp(O(k^2+k\log (1/\eps))\log(1/\delta)$  total evolution time. 
\end{theorem}

The learning algorithm of \cref{theo:locallearning} has two stages. In the first stage one samples from the Pauli distribution of $U(\eps)$, as in the testing algorithm, and from that one can detect which are the big Pauli coefficients of $H$. In the second stage we learn those big Pauli coefficients using the SWAP test on $U(\eps)$, as proposed by Montanaro and Osborne \cite[Lemma 24]{montanaro2008quantum}. One can ensure that the coefficients not detected as big in the first stage of the algorithm can be neglected. To do that we borrow the ideas of Eskenazis and Ivanisvili for the classical low-degree learning problem \cite{eskenazis2022learning}, combined with the non-commutative Bohnenblust-Hille inequality proved by Huang, Chen and Preskill \cite{huang2023learning}, and improved by Volberg and Zhang \cite{volberg2023noncommutative}. 

Most  previous work on Hamiltonian learning is done under the distance induced by the supremum norm of the Pauli spectrum and with extra constraints apart from locality \cite{Silva2011Practical, Bairey2019Learning, Zubida2021Optimal, haah2022optimal, wilde2022scalably, yu2023robust, caro2023learning, Dutkiewicz.2023, huang2023heisenberg, li2023heisenberglimited, möbus2023dissipationenabled, franca2024efficient, Gu2022Practical}. When transformed into learning algorithms under the finer distance induced by the  2-norm of the Pauli spectrum, these proposals yield complexities that depend polynomially on $n^k$ and only work for a restricted family of $k$-local Hamiltonians. A work that explicitly considers the problem of learning under the 2-norm of the Pauli spectrum is the one of Castaneda and Wiebe \cite{castaneda2023Hamiltonian}. However, their results require query access to the inverse time evolution operator $e^{-itH}$ and yield complexities of order $O(n^k)$. Hence, to the best of our knowledge, \cref{theo:locallearning} is the first learning result that has no dependence on $n$ when considering the 2-norm of Pauli spectrum and works for any kind of $k$-local Hamiltonian. Regarding the tightness of its query complexity, the $\exp(\Omega(k))$ query lower bound of Bluhm, Caro and Oufkir \cite{bluhm2024Hamiltonianv1} to learn $k$-qubit Hamiltonians shows that our algorithm cannot be improved much. Closing the gap between this lower bound and our $\exp(O(k^2))$ query upper bound remains an intriguing open problem. 

\section{Preliminaries}
In this section we collect a few well-known facts that we will repeatedly use in our proofs. Given an $n$-qubit operator $A=\sum_{x}a_x\sigma_x$, Parseval's identity states that its normalized Frobenius norm equals the $2$-norm of its Pauli spectrum. We will denote both by $\norm{A}_2$,
$$\norm{A}_2=\sqrt{\frac{\Tr [A^\dagger A]}{2^n}}=\sqrt{\sum_{x\in\{0,1,2,3\}^2}|a_x|^2}.$$
Given $x\in\{0,1,2,3\}^n$, we define $|x|$ as the number of sites where $x$ does not take the value 0, $A_{>k}$ as $\sum_{|x|>k}a_x\sigma_x$ and $A_{\leq k}$ as $\sum_{|x|\leq k}a_x\sigma_x$. From the formulation of the 2-norm in terms of the Pauli coefficients it follows that $\norm{A_{>k}}_2\leq \norm{A}_2$, while from its formulation as the normalized Frobenius norm one has that $\norm{A}_2\leq \norm{A}_\infty.$  We recall that $\norm{A}_{\infty}$ is the biggest singular value of $A$.
We note that the distance of a Hamiltonian $H$ from the space of $k$-local Hamiltonians is given by $\norm{H_{>k}}_2$, as $H_{\leq k}$ is the  $k$-local Hamiltonian closest to $H$. 

It follows from Parseval's identity that if $U$ is a unitary, then $\sum|u_x|^2=1.$ In other words, $(|u_x|^2)_x$ is a  probability distribution. Applying $U\otimes \Id_{2^n}$ to $n$ EPR pairs (i.e., preparing the Choi-Jamiolkowski state of $U$) and measuring in the Bell basis allows one to sample from this distribution, because $$U\otimes\Id_{2^n}\ket{\mathrm{EPR}_n}=\sum_{x\in\{0,1,2,3\}^n}u_x \otimes_{i\in [n]}(\sigma_{x_i}\otimes \Id_2\ket{\mathrm{EPR}}),$$
and the Bell states can be written as $\sigma_x\otimes\Id_{2}\ket{\mathrm{EPR}}$ for $x\in\{0,1,2,3\}$.

We will also use that given a Hamiltonian $H$ with $\norm{H}_\infty\leq 1$, the Taylor expansion of the exponential allows us to approximate the time evolution operator as 
$$ 
U(t)=e^{-itH}=\Id_{2^n}-itH+ct^2 R_2(t)
$$
for $t\leq 1/2,$ where the second order remainder $R_2(t)$ is bounded $\norm{R_2(t)}_\infty\leq 1$ and $c>0$ is a universal constant. 

\section{Testing locality}
In this section we prove \cref{theo:localitytesting}. First, we prove a claim regarding the discrepancy on the weights of non-local terms of the short-time evolution operator for close-to-local and far-from-local Hamiltonians. 
\begin{claim}\label{theo:testingdicotomy}
    Let $0\leq \eps_1<\eps_2$. Let $\alpha=(\eps_2-\eps_1)/(3c)$ and $H$ be an $n$-qubit Hamiltonian with $\norm{H}_{\infty}\leq 1$. We have that if $H$ is $\eps_1$-close $k$-local, then $$\norm{U(\alpha)_{>k}}_2\leq (\eps_2-\eps_1)\frac{2\eps_1+\eps_2}{9c},$$ and if $H$ is $\eps_2$-far from being $k$-local, then  $$\norm{U(\alpha)_{>k}}_2\geq (\eps_2-\eps_1)\frac{\eps_1+2\eps_2}{9c}.$$
\end{claim}
\begin{proof}To save on notation, we set $U=U(\alpha)$ and $R=R_2(\alpha)$. First, assume that $H$ is $\eps_1$-close $k$-local. Then  $$\norm{U_{>k}}_2\leq \alpha\norm{H_{>k}}_2+ c\alpha^2\norm{R_{>k}}_2\leq \frac{\eps_2-\eps_1}{3c}\eps_1+c\left(\frac{\eps_2-\eps_1}{3c}\right)^2 =(\eps_2-\eps_1)\frac{2\eps_1+\eps_2}{9c},$$
    where in the first inequality we have used the triangle inequality and the Taylor expansion, and in the second that $H$ is $\eps_1$-close to $k$-local and that $\norm{R_{>k}}_2\leq \norm{R}_2\leq 1$ because $\norm{R}_2\leq \norm{R}_\infty\leq 1.$ Now, assume that $H$ is $\eps_2$-far from being $k$-local. Then $$\norm{U_{>k}}_2\geq \alpha\norm{H_{>k}}_2-c\alpha^2\norm{R_{>k}}_2\geq \frac{\eps_2-\eps_1}{3c}\eps_2-c\left(\frac{\eps_2-\eps_1}{3c}\right)^2 \geq (\eps_2-\eps_1)\frac{\eps_1+2\eps_2}{9c},$$
    where in first inequality we have used again Taylor expansion and the triangle inequality, and in the second the fact that being $\eps_2$-far from $k$-local implies that $\norm{H_{>k}}_2\geq \eps_2.$
\end{proof}

\begin{proof}[ of  \cref{theo:localitytesting}]
    Applying $U(\alpha)\otimes \Id_{2^n}$ to $\ket{\mathrm{EPR}_n}$ and measuring in the Bell basis allows one to sample from $(|U(\alpha)_x|^2)_{x}$. Thus, by the Hoeffding bound, with $O(1/(\eps_2-\eps_1)^8\cdot \log(1/\delta))$ queries to $U(\alpha)\otimes\Id_{2^n}$  one can estimate $\sum_{|x|>k}|U(\alpha)_x|^2$ up to an error $((\eps_2-\eps_1)^2/(18c))^2$ with success probability $1-\delta$. Taking $\alpha=(\eps_2-\eps_1)/(3c)$, thanks to \cref{theo:testingdicotomy}, this is enough for testing $k$-locality.
\end{proof}

\begin{remark}
    The algorithm of \cref{theo:localitytesting} also works to test any property defined by a set of Pauli strings $S\subseteq \{0,1,2,3\}^n$, i.e., to test whether $\sqrt{\sum_{x\notin S}h_x^2}\leq \eps_1$  or $\sqrt{\sum_{x\notin S}h_x^2}\geq \eps_2$. Also, a union bound allows us to simultaneously test $M$ properties defined by sets $S_1,...,S_M$ by  paying a factor of $\log M$ in the query complexity and total evolution time. 
\end{remark}

\section{Learning local Hamiltonians}
In this section we prove \cref{theo:locallearning}. To do that, we need the non-commutative Bohnenblust-Hille inequality by Volberg and Zhang \cite{volberg2023noncommutative}, which they used to give an algorithm to learn local observables. 

\begin{theorem}[Non-Commutative Bohnenblust-Hille inequality]\label{theo:NCBH}
    Let $H=\sum_x h_x\sigma_x$ be a $k$-local Hamiltonian with $\norm{H}_{\infty}\leq 1$. Then, there is a universal constant $C$ such that $$\sum_{x\in\{0,1,2,3\}^n}|h_x|^{\frac{2k}{k+1}}\leq C^{k}.$$
\end{theorem}

\begin{proof}[ of \cref{theo:locallearning}]
    Let $\alpha,\beta,\gamma$  be fixed later. Let $U$ be the evolution operator $U(\alpha)$ at time $\alpha$ and $R$ be the Taylor remainder $R_2(\alpha)$. Our learning algorithm has two stages, the first to detect the big Pauli coefficients and the second to learn those big Pauli coefficients.
    
    \textbf{First stage of the algorithm: detect the big Pauli coefficients.} In this stage we prepare $U\otimes \Id_{2^n}\ket{\mathrm{EPR}_n}$ and measure in the Bell basis $O(\gamma^4\log(1/\delta))$ times. This way  we  sample $O(\gamma^4\log(1/\delta))$ times from $(|u_x|^2)_x$. By \cite[Theorem 9]{canonne2020short} the empirical distribution $|u_x'|^2$ obtained from these samples approximates $|u_x|^2$ up to error $\gamma^2$ for every $x\in\{0,1,2,3\}^n$ with success probability $\geq 1-\delta$. Let $S_{\gamma}=\{x: | u_x'|>\gamma\}\setminus\{0^n\}$ be the set of big Pauli coefficients. Note that if $x$ does not belong to $S_{\gamma}$, then $u_x$ is indeed small since 
    \begin{equation}\label{eq:estimate1}
        x\notin S_\gamma \land x\neq 0 \implies | u_x|\leq | u_x'|+| u_x- u_x'|\leq 2\gamma.
    \end{equation}
    Also note that $u_0=1-i\alpha h_0+c\alpha^2 R_{0}$ and that $ u_x=-i\alpha h_x+c\alpha^2 R_{x}$ for every $x\neq 0^n$. As $\norm{R}_{\infty}\leq 1$, we have that  $\norm{c\alpha^2R}_\infty^2\leq c^2\alpha^4$. As $\norm{R}_{2}\leq \norm{R}_{\infty}$, this implies that  
    \begin{equation}\label{eq:estimate4}
        |(u_0-1)+i\alpha h_0|^2+\sum_{x\neq 0}| u_x+i\alpha h_x|^2\leq c^2\alpha^4
    \end{equation}
    and in particular for every $x\neq 0$
    \begin{equation}\label{eq:estimate2}
         |u_x+i\alpha h_x|\leq c\alpha^2.
    \end{equation}
   Putting \cref{eq:estimate1,eq:estimate2} together it follows that
     \begin{equation}\label{eq:estimate3}
        x\notin S_\gamma \land x\neq 0 \implies | h_x|\leq \alpha^{-1}[| u_x|+| u_x+i\alpha h_x|]\leq  \alpha^{-1}(2\gamma+c\alpha^2).
    \end{equation}
    Moreover, as $\sum_{x}| u_x'|^2=1,$ we have that $|S_\gamma|\leq \gamma^{-2}.$ 
    
    \textbf{Second stage of the algorithm: learn the big Pauli coefficients.} Montanaro and Osborne proposed a simple primitive to estimate a given Pauli coefficient of a unitary up to error $\beta$ with success probability $\geq 1-\delta$ by making  $O((1/\beta)^2\log(1/\delta))$ queries \cite[Lemma 24]{montanaro2008quantum}. 
    We use this primitive to learn the Pauli coefficients of $u_x$ up to error $\beta$ with success probability $\geq 1-\delta$ for every $x\in S_\gamma\cup\{0^n\}$. As,  $|S_\gamma|\leq \gamma^{-2}$, by a union bound, this stage requires $O(\beta^{-2}\gamma^{-2}\log(1/(\gamma^2\delta)))$ queries to $U$. Let $u''_x$ be these estimates. We output $H''=\mathrm{Re}(i\alpha^{-1}(u''_{0}-1))\sigma_0+\sum_{x\in S_\gamma} \mathrm{Re}(i\alpha^{-1}u_x'')\sigma_x$ as our approximation of $H$ (we take the real part to ensure that $H''$ is self-adjoint). 

    \textbf{Correctness of the algorithm.} We claim that with an appropriate choice of the parameters $H''$ is a good approximation of $H$. First, as the Pauli coefficients of $H$ are real we have that 
    \begin{align*}
        \norm{H-H''}_2^2&=\alpha^{-2}|\mathrm{Re}(i(u''_{0}-1))-\alpha h_0|^2+\alpha^{-2}\sum_{x\in S_\gamma}| \mathrm{Re}(iu_x'')-\alpha h_x|^2+\sum_{x\notin S_\gamma}|h_x|^2\\
        &\leq\underbrace{\alpha^{-2}|(u_0''-1)+i\alpha h_0|^2+\alpha^{-2}\sum_{x\in S_\gamma}| u_x''+i\alpha h_x|^2}_{\mathrm{(I})}+\underbrace{\sum_{x\notin S_\gamma}|h_x|^2}_{\mathrm{(II})}.
    \end{align*}
    Second, we give an upper bound to term $($I$)$: 
    \begin{align*}
        (\mathrm{I})&\leq 2\alpha^{-2}[|(u_0-1)+i\alpha h_0|^2+|u_{0}''-u_{0}|^2]+2\alpha^{-2}\sum_{x\in S_\gamma}[| u_x+i\alpha h_x|^2+| u_x''-u_x|^2]\\
        &\leq 2c^2\alpha^2+2\alpha^{-2}\beta^2(\gamma^{-2}+1)
    \end{align*}
    where in the first step we have used the triangle inequality and that $(a+b)^2\leq 2(a^2+b^2)$; and in the second step \cref{eq:estimate4} to upper bound the error due to approximating $h_x$ from $u_x$, and the learning guarantees of the second stage of the algorithm and that $|S_{\gamma}|\leq \gamma^{-2}$ to upper bound the error due to approximating $u_x$ by $u_x''$.  
    Third, we upper bound term $($II$)$ 
    \begin{align*}
        (\mathrm{II})&\leq\max_{x\notin S_\gamma}\{|h_x|^{2/k+1}\}\sum_{x\notin S_\gamma}|h_x|^{2k/k+1}\leq (2\gamma\alpha^{-1}+c\alpha)^{2/k+1} C^{k},
    \end{align*}
    where in the first step we have used that  $2=2/(k+1)+2k/(k+1)$, and in the second step \cref{theo:NCBH} and \cref{eq:estimate4}. Finally, if we put everything together we get
    \begin{align*}
        \norm{H-H''}_2^2&\leq 2c^2\alpha^2+2\alpha^{-2}\beta^2(\gamma^{-2}+1)+(2\gamma\alpha^{-1}+c\alpha)^{2/k+1} C^{k}.
    \end{align*}
    Thus, if we take $\alpha=\eps^{k+1}C^{-k(k+1)/2},\ \gamma=\alpha^{2},$ and $ \beta=\alpha^{3}\eps$ it follows that $\norm{H-H''}_2^2\leq O(\eps^2),$ as desired. 
    
    \textbf{Query complexity and total evolution time.} Taking both stages of the algorithm into account we make $O(\gamma^{-4}\log(1/\delta)+\gamma^{-2}\beta^{-2}\log(1/(\gamma^{2}\delta)))$ queries to $U(\alpha)$. Hence, by taking $\alpha,\beta$ and $\gamma$ as above, one obtains the claimed query complexity and total evolution time.
\end{proof}

\paragraph{Acknowledgements.} I thank Srinivasan Arunachalam and Matthias Caro for detailed comments and advice about presentation. I thank Amira Abbas and Jop Briët for useful comments, discussions and encouragement. I thank  Andreas Bluhm, Arkopal Dutt and Aadil Oufkir for useful comments and discussions. This research was supported by the European Union’s Horizon 2020 research and innovation programme under the Marie Sk{\l}odowska-Curie grant agreement no. 945045, and by the NWO Gravitation project NETWORKS under grant no. 024.002.003.

\bibliographystyle{alphaurl}
\bibliography{Bibliography}
\end{document}